\documentclass[11pt]{article}
\usepackage[height=220mm,width=150mm]{geometry}
\usepackage{array}
\usepackage{color}
\usepackage{mathrsfs}
\usepackage{amssymb}
\usepackage{amsmath}
\usepackage{amsthm}
\allowdisplaybreaks
\usepackage{graphicx}
\usepackage{subfigure}
\usepackage{tikz}
\usetikzlibrary{shapes.arrows,chains,positioning}
\newtheorem{Theorem}{Theorem}[section]
\newtheorem{Lemma}[Theorem]{Lemma}
\newtheorem{Definition}[Theorem]{Definition}
\newtheorem{Proposition}[Theorem]{Proposition}

\usepackage{epsfig,psfrag}

\begin{document}
%
\title{Nonuniform Sampling for Random Signals Bandlimited in the Linear Canonical Transform Domain\thanks{This work was partially supported by the
National Natural Science Foundation of China (11525104, 11531013 and 11371200).}}

\author{Haiye Huo$^a$, Wenchang Sun$^b$\\
$^a$ Department of Mathematics, School of Science, Nanchang University,\\ Nanchang~330031, Jiangxi, China \\
\mbox{} \\
$^b$ School of Mathematical Sciences and LPMC, Nankai University,\\ Tianjin~300071, China \\
\mbox{} \\
Emails: hyhuo@ncu.edu.cn; sunwch@nankai.edu.cn}

\date{}
\maketitle

\textit{Abstract}.\,\,
In this paper, we mainly investigate the nonuniform sampling for random signals which are bandlimited in the linear canonical transform (LCT) domain.  We show that the nonuniform sampling for a random signal bandlimited in the LCT domain
is equal to the uniform sampling in the sense of second order statistic characters after a pre-filter in the LCT
domain. Moreover, we propose an approximate recovery approach for nonuniform sampling of random signals bandlimited in the LCT domain.
Furthermore, we study the mean square error of the nonuniform sampling. Finally, we do some simulations to verify the correctness of
our theoretical results.


\section{Introduction}\label{sec I}

Sampling is very fundamental in signal processing, as it provides an effective way to connect the analogue signals and digital signals. Since Shannon \cite{Shannon1949} introduced the concept of sampling theorem in 1949, the sampling theorem has been widely studied in various academic fields.
In particular, uniform sampling theorems for deterministic signals or random signals, which are bandlimited in the Fourier domain, fractional Fourier transform domain, or linear canonical transform (LCT) domain, have been intensely studied in literatures \cite{BM2010,Brown1978,HS2015,Shannon1949,SLS2012,SLS2016,Stern2007,SZ2002,TZW2011,WRL2011,WL2014,XQ2014,XS2013,Zhang2016}.

In practice, we might only obtain nonuniform samples, for instance, in the areas of geophysics, biomedical imaging, or communication theory \cite{Eng2007,Leow2010,Senay2009}.
Therefore, nonuniform sampling has aroused much more attention on the theoretical and practical sides in the literature.
There are many kinds of approaches for recovering the original signals from their nonuniform samples.
For example, Yao and Thomas \cite{YT1967} derived the reconstruction formula for bandlimited signals from their nonuniform samples by using the Lagrange interpolation functions. However, since the Lagrange interpolation functions often have distinct formats at different sampling times, it is very complicated to recover the bandlimited signals by utilizing Lagrange interpolation functions. Many other approaches have been presented to solve this problem.
Time-warping technique was used in \cite{Papoulis1966} for recovering bandlimited signals from their jittered samples. In \cite{Selva2009}, the author made some revision on traditional Lagrange interpolation functions, in order to improve the accuracy of recovering bandlimited signals from their nonuniform samples. Due to the perfect recovery of bandlimited signals from their nonuniform samples with sinc interpolation, Maymon and Oppenheim \cite{MO2011} proposed a class of approximate recovery approaches for bandlimited signals from their nonuniform samples by utilizing sinc interpolation functions. Furthermore, Xu, Zhang and Tao \cite{XZT2016} generalized the results mentioned in \cite{MO2011} from traditional Fourier domain to fractional Fourier transform domain. In order to learn more information on nonuniform sampling, we refer the readers to \cite{AG2001,Balak1962,CGE2014,FG1992,Marvasti2012,TLWA2008,VB2000,Yen1956}.

For random signals which are bandlimited in the LCT domain, there exist few results on sampling theorems.
In \cite{HS2015}, based on the framework of LCT auto-correlation function and power spectral density,
we investigated the uniform sampling theorem and multichannel sampling theorem for random signals which are bandlimited in LCT domain.
In this paper, we derive the relationship between the LCT auto-power spectral densities of the inputs and outputs. In addition, we study the nonuniform sampling for random signals bandlimited in the LCT domain and give an approximate recovery approach with sinc interpolation functions.
Moreover, we investigate the error estimate of nonuniform sampling for random signals bandlimited in the LCT domain in the mean square sense.
Finally, some simulations are carried out to illustrate the effectiveness of our methods.

The rest of the paper is presented as follows. In Section \ref{sec P}, we first introduce the concepts of the LCT, the LCT correlation function, and the LCT power spectral density. Then, we show the connection between the LCT auto-power spectral density of the inputs and outputs. In Section \ref{sec N}, we study the nonuniform sampling, its approximate recovery method, the corresponding reconstruction error for random signals bandlimited in the LCT domain in the mean square sense. Moreover, we analyze the performances of our theoretical results by simulation. In Section \ref{sec C}, we conclude the paper.

\section{Preliminaries}\label{sec P}

\subsection{The Linear Canonical Transform}
\begin{Definition}\label{def:LCT}
The LCT of a signal $f(t)\in L^2(\mathbb{R})$ is denoted by \cite{HS2015}
\begin{equation}\label{def:LCT:L0}
\mathcal{L}_{A}\{f(t)\}(u)=
    \begin{cases}
    \int_{-\infty}^{+\infty}f(t)\sqrt{\frac{1}{j2\pi b}}e^{j\frac{a}{2b}t^{2}-j\frac{1}{b}ut+j\frac{d}{2b}u^{2}}{\rm{d}}t, & b\ne0,\\
     \sqrt{d}e^{j\frac{cd}{2}u^{2}}f(du), & b=0,\\
     \end{cases}
\end{equation}
where $A=
\left(
\begin{array}{cc}
a & b \\
c & d \\
\end{array}
\right)$,
and parameters $a,\;b,\;c,\;d\in \mathbb{R}$ satisfy $ad-bc=1$.
\end{Definition}

Since the LCT is a Chirp multiplication operator when $b=0$, we assume without loss of the generality that $b>0$ in the rest of the paper.
From (\ref{def:LCT:L0}), we can easily derive the connection between the LCT and the Fourier transform as follows:
\begin{equation}\label{eq:relation:0}
\mathcal{L}_{A}\{f(t)\}(u)=\sqrt{\frac{1}{j2\pi b}}e^{j\frac{d}{2b}u^2}\mathcal{F}(f(t)e^{j\frac{a}{2b}t^2})\big(\frac{u}{b}\big),
\end{equation}
where the Fourier transform of  $f(t)$ is defined by
\begin{equation}\label{eq:relation:1}
\mathcal{F}(f)(u)=\int_{-\infty}^{+\infty}f(t)e^{-jut}{\rm{d}}t.
\end{equation}

\subsection{The LCT Power Spectral Density}
In this paper, we consider a special class of random signals that is wide sense stationary.
Given a probability space $(\Omega, \mathscr{F}, \mathbb{P})$, a stochastic process $x(t)$ is called stationary in a wide sense, if its mean is zero, its second moment is finite, and its auto-correlation function
\begin{equation}\label{correlation}
R_{xx}(t+\tau,t)=\mathbb{E}[x(t+\tau)x^*(t)]
\end{equation}
is independent of $t\in \mathbb{R}$, where $\mathbb{E}$ denotes mathematical expectation, and $x^*$ stands for the complex conjugate of $x$.
Two stochastic processes $x(t)$ and $y(t)$ are said to be jointly stationary in a wide sense, if $x(t)$ and $y(t)$ are both wide sense stationary,
and their cross-correlation function
\begin{equation}\label{cross-corre}
R_{xy}(t+\tau,t)=\mathbb{E}[x(t+\tau)y^*(t)]
\end{equation}
is independent of $t\in \mathbb{R}$.

Next, we introduce the LCT auto-correlation function, the LCT cross-correlation function, the LCT auto-power spectral density and the LCT cross-power spectral density as follows.

\begin{Definition}\label{Def:correlation}
Given random signals $x(t)$ and $y(t)$, the LCT auto-correlation function of $x(t)$ is defined by
\begin{equation}\label{auto-correlation}
R_{xx}^{A}(t_1,t_2)=\mathbb{E}[x(t_1)x^{*}(t_2)e^{j\frac{a}{b}t_2(t_1-t_2)}]=R_{xx}(t_1,t_2)e^{j\frac{a}{b}t_2(t_1-t_2)},
\end{equation}
and the LCT cross-correlation function of $y(t)$ and $x(t)$ is defined by
\begin{equation}\label{cross-correlation}
R_{yx}^{A}(t_1,t_2)=\mathbb{E}[y(t_1)x^{*}(t_2)e^{j\frac{a}{b}t_2(t_1-t_2)}]=R_{yx}(t_1,t_2)e^{j\frac{a}{b}t_2(t_1-t_2)}.
\end{equation}
\end{Definition}

One can see that if $\tilde{x}(t)=x(t)e^{j\frac{a}{2b}t^2}$ is stationary, that is,
\begin{equation}\label{LCT:A0}
R_{\tilde{x}\tilde{x}}(t_1,t_2)=R_{\tilde{x}\tilde{x}}(\tau),
\end{equation}
where $\tau=t_1-t_2,$ then the function $R_{xx}^{A}(t_1,t_2)$ also depends only on $\tau$.
In fact,
\begin{eqnarray}\label{LCT:A1}
R_{xx}^{A}(t_1,t_2)&=&\mathbb{E}[x(t_1)x^{*}(t_2)e^{j\frac{a}{b}t_2(t_1-t_2)}]\nonumber\\
&=&\mathbb{E}[x(t_1)e^{j\frac{a}{2b}t_1^2}x^{*}(t_2)e^{-j\frac{a}{2b}t_2^2}]e^{-j\frac{a}{2b}(t_1-t_2)^2}\nonumber\\
&=&R_{\tilde{x}\tilde{x}}(t_1,t_2)e^{-j\frac{a}{2b}(t_1-t_2)^2}\nonumber\\
&=&R_{\tilde{x}\tilde{x}}(\tau)e^{-j\frac{a}{2b}\tau^2}.
\end{eqnarray}

\begin{Definition}\label{Def:LCT spectrum}
Given random signals $x(t)$ and $y(t)$, and two parameters
$A=\left(
  \begin{array}{cc}
    a & b \\
    c & d \\
  \end{array}
\right)$,
$A^{\prime}=\left(
             \begin{array}{cc}
               a & -b \\
               -c & d \\
             \end{array}
           \right).$
The LCT auto-power spectral density of $x(t)$  is defined by
\begin{eqnarray}
P_{xx}^{A}(u)
&=&\sqrt{\frac{1}{-j2\pi b}}e^{-j\frac{d}{2b}u^2}\mathcal{L}_{A}\{R_{xx}^{A}(\tau)\}(u)\nonumber\\
&=&\sqrt{\frac{1}{j2\pi b}}e^{j\frac{d}{2b}u^2}\mathcal{L}_{A^{\prime}}\{R_{xx}^{A,2}(\tau)\}(-u),\label{Def:Spec:1}
\end{eqnarray}
and the LCT cross-power spectral density of $y(t)$ and $x(t)$  is defined by
\begin{eqnarray}
P_{yx}^{A}(u)
&=&\sqrt{\frac{1}{-j2\pi b}}e^{-j\frac{d}{2b}u^2}\mathcal{L}_{A}\{R_{yx}^{A}(\tau)\}(u)\nonumber\\
&=&\sqrt{\frac{1}{j2\pi b}}e^{j\frac{d}{2b}u^2}\mathcal{L}_{A^{\prime}}\{R_{yx}^{A,2}(\tau)\}(-u),\label{Def:Spec:2}
\end{eqnarray}
where
\begin{equation}\label{Def:Spec:3}
R_{xx}^{A,2}(t_1,t_2)=\mathbb{E}[x(t_1)x^{*}(t_2)e^{j\frac{a}{b}t_1(t_1-t_2)}]=R_{xx}(t_1,t_2)e^{j\frac{a}{b}t_1(t_1-t_2)}
\end{equation}
and
\begin{equation}\label{Def:Spec:4}
R_{yx}^{A,2}(t_1,t_2)=\mathbb{E}[y(t_1)x^{*}(t_2)e^{j\frac{a}{b}t_1(t_1-t_2)}]=R_{yx}(t_1,t_2)e^{j\frac{a}{b}t_1(t_1-t_2)}.
\end{equation}
\end{Definition}

It follows from (\ref{def:LCT:L0}) and (\ref{Def:Spec:1}) that
\begin{equation}\label{spectrum:A1}
R_{xx}^A(\tau)=\int_{-\infty}^{+\infty}P_{xx}^A(u)e^{-j\frac{a}{2b}\tau^2+j\frac{1}{b}u\tau}{\rm{d}}u.
\end{equation}

In \cite{HS2015,DTW2006}, a model of LCT multiplicative filter has been introduced as in Fig.~\ref{Fig.1}, where $X(u)=\mathcal{L}_{A}\{x(t)\}(u)$, $Y(u)=X(u)H(u)$, and the output function $y(t)$ is given by
\begin{equation}\label{multi filter}
y(t)=\mathcal{L}_{A^{-1}}\{Y(u)\}(t)=\mathcal{L}_{A^{-1}}\{X(u)H(u)\}(t).
\end{equation}

With the LCT multiplicative filter described in Fig.~\ref{Fig.1}, we can obtain the relationship between the LCT auto-power spectral density  $P_{xx}^{A}(u)$ and cross-power spectral density $P_{yx}^{A}(u)$.

\begin{Proposition}\cite[Theorem 2.3]{HS2015}\label{Proposition:Spec}
Let random signals $x(t)$ and $y(t)$ be the input and output of the LCT multiplicative filter, and the transfer function $H(u)$ satisfy
\begin{equation}\label{Proposition:Spec1}
h(t) =\frac{1}{\sqrt{2\pi}}\int_{-\infty}^{+\infty}H(u)e^{jut/b}\textrm{d}u.
\end{equation}
Then,
\begin{equation}\label{Proposition:Spec2}
P_{yx}^{A}(u)=H(u)P_{xx}^{A}(u).
\end{equation}
\end{Proposition}

Similarly, we can get the connection between $P_{yy}^{A}(u)$ and $P_{xx}^{A}(u)$ as follows.
\begin{Theorem}\label{Thm:Spectrum}
Let random signals $x(t)$ and $y(t)$ be the input and output of the LCT multiplicative filter, and the transfer function $H(u)$ satisfy
\begin{equation}\label{Thm:Spec1}
h(t) =\frac{1}{\sqrt{2\pi}}\int_{-\infty}^{+\infty}H(u)e^{jut/b}\textrm{d}u.
\end{equation}
Then,
\begin{equation}\label{power spectrum density}
P_{yy}^{A}(u)=|H(u)|^2P_{xx}^{A}(u).
\end{equation}
\end{Theorem}

\begin{proof}
By (\ref{multi filter}), we have
\begin{equation}\label{hhh}
y(t)=\mathcal{L}_{A^{-1}}\{\mathcal{L}_{A}\{x(t)\}(u)H(u)\}(t).
\end{equation}
Using the convolution theorem (\cite[Theorem 1]{DTW2006}), we can rewrite (\ref{hhh}) as
\begin{equation}\label{Thm:Spectrum:A1}
y(t)=\frac{1}{\sqrt{2\pi}b}\int_{-\infty}^{+\infty}x(t-\tau)e^{j\frac{a}{2b}(\tau^{2}-2t\tau)}h(\tau)\textrm{d}\tau.
\end{equation}
Thus,
\begin{align}
&R_{yy}(t_{1},t_{2})=\mathbb{E}[y(t_{1})y^*(t_{2})]\nonumber\\
=&\frac{1}{\sqrt{2\pi}b}\int_{-\infty}^{+\infty}h^*(u)e^{-j\frac{a}{2b}(u^2-2t_{2}u)}R_{yx}(t_{1},t_{2}-u)\textrm{d}u\nonumber\\
=&\frac{1}{\sqrt{2\pi}b}\int_{-\infty}^{+\infty}h^*(u)e^{j\frac{a}{b}t_{1}u}[R_{yx}(t_{1},t_1-(\tau+u))
   e^{-j\frac{a}{2b}(u^2+2\tau u)}]\textrm{d}u,\label{Thm:Spectrum:A3}
\end{align}
where $\tau=t_{1}-t_{2}$.
Combining (\ref{Def:Spec:4}) and (\ref{Thm:Spectrum:A3}), we get
\begin{align}
&R_{yy}^{A,2}(\tau)\nonumber\\
=&R_{yy}(t_1,t_1-\tau)e^{j\frac{a}{b}t_1\tau}\nonumber\\
=&\frac{1}{\sqrt{2\pi}b}\int_{-\infty}^{+\infty}h^*(u)e^{j\frac{a}{b}t_1(\tau+u)}R_{yx}(t_{1},t_1-(\tau+u))
   e^{-j\frac{a}{2b}(u^2+2\tau u)}\textrm{d}u\nonumber\\
=&\frac{1}{\sqrt{2\pi}b}\int_{-\infty}^{+\infty}h^*(u)R_{yx}^{A,2}(\tau+u)e^{-j\frac{a}{2b}(u^2+2\tau u)}\textrm{d}u\nonumber\\
=&\frac{1}{\sqrt{2\pi}b}\int_{-\infty}^{+\infty}h^*(-u)R_{yx}^{A,2}(\tau-u)e^{j\frac{a}{2(-b)}(u^2-2\tau u)}\textrm{d}u.\label{Thm:Spectrum:A5}
\end{align}
Therefore, we obtain
\begin{equation}\label{Thm:Spectrum:A6}
\mathcal{L}_{A^\prime}\{R_{yy}^{A,2}(\tau)\}(u)=H^*(-u)\mathcal{L}_{A^{\prime}}\{R_{yx}^{A,2}(\tau)\}(u).
\end{equation}
Substituting (\ref{Def:Spec:1}) and (\ref{Def:Spec:2}) into (\ref{Thm:Spectrum:A6}), we have
\begin{equation}\label{Thm:Spectrum:A7}
P_{yy}^{A}(u)=H^*(u)P_{yx}^{A}(u).
\end{equation}
It follows from (\ref{Proposition:Spec2}) that
\[
P_{yy}^{A}(u)=|H(u)|^2P_{xx}^{A}(u).
\]
This completes the proof.
\end{proof}

\section{Nonuniform Sampling and Error Estimate for Random Signals Bandlimited in the LCT Domain}\label{sec N}

In this section, we study the  nonuniform sampling and error estimate for random signals  which are bandlimited in the LCT domain. First, we introduce the definition.

\begin{Definition}\cite{HS2015}\label{def:bandlimited}
We call a random signal $x(t)$  bandlimited in the LCT domain, if its LCT power spectral density $P_{xx}^{A}(u)$ satisfies
\begin{equation}\label{def:bandlimited1}
P_{xx}^{A}(u)=0,\qquad |u|>u_{r},
\end{equation}
where $u_{r}$ is the bandwidth.
\end{Definition}

Before stating our main results, we present a lemma that is useful in the following.

\begin{Lemma}\label{Lem:L1}
Assume that a random signal $x(t)$ is bandlimited in the LCT domain with bandwidth $u_r$, and $\tilde{x}(t)=x(t)e^{j\frac{a}{2b}t^2}$ is a wide sense stationary process. Then, $\tilde{x}(t)$ is bandlimited in the Fourier domain with bandwidth $u_r/b.$
\end{Lemma}

\begin{proof}
Since $\tilde{x}(t)$ is stationary in a wide sense, it follows from (\ref{LCT:A1}) and (\ref{Def:Spec:1}) that
\begin{eqnarray}
P_{xx}^{A}(u)
&=&\sqrt{\frac{1}{-j2\pi b}}e^{-j\frac{d}{2b}u^2}\mathcal{L}_{A}\{R_{xx}^{A}(\tau)\}(u)\nonumber\\
&=&\sqrt{\frac{1}{-j2\pi b}}e^{-j\frac{d}{2b}u^2}\mathcal{L}_{A}\{R_{\tilde{x}\tilde{x}}(\tau)e^{-j\frac{a}{2b}\tau^2}\}(u)\nonumber\\
&=&\sqrt{\frac{1}{-j2\pi b}}\sqrt{\frac{1}{j2\pi b}}\Big[\int_{-\infty}^{+\infty}R_{\tilde{x}\tilde{x}}(\tau)
    e^{-j\frac{1}{b}u\tau}{\rm{d}}\tau\Big]\nonumber\\
&=&\frac{1}{2\pi b}P_{\tilde{x}\tilde{x}}(\frac{u}{b}).\label{Lem:L2}
\end{eqnarray}
Note that
\begin{equation}\label{Lem:L3}
P_{xx}^{A}(u)=0,\qquad |u|>u_{r}.
\end{equation}
We have
\[
P_{\tilde{x}\tilde{x}}(u)=0,\qquad |u|>\frac{u_{r}}{b}.
\]
This completes the proof.
\end{proof}

\subsection{Nonuniform Sampling}\label{subsec:N1}

First, we restate a nonuniform sampling theorem for deterministic signals which are bandlimited in the LCT domain, as mentioned in \cite{TLWA2008}.

\begin{Proposition}\cite[Theorem 4]{TLWA2008}\label{nonuniform:n1}
Suppose that a deterministic signal $f(t)$ is bandlimited in the LCT domain with bandwidth $u_r$. If
\begin{equation}\label{Prop:N1}
|t_n-n\frac{b\pi}{u_r}|\le D<\frac{b\pi}{4u_r},
\end{equation}
then the function $f(t)$ can be perfectly recovered by its samples $f(t_n)$ with the following formula,
\begin{equation}\label{Prop:N2}
f(t)=e^{-j\frac{a}{2b}t^2}\sum_{-\infty}^{+\infty}f(t_n)e^{j\frac{a}{2b}t_n^2}\frac{G(t)}{G^{\prime}(t_n)(t-t_n)},
\end{equation}
where
\begin{eqnarray*}
G(t)&=&e^{\alpha t}(t-t_0)\prod\limits_{n\ne 0}\Big(\frac{1-t}{t_n}\Big)e^{t/t_n},\\
\alpha&=&\sum\limits_{n\ne 0}\frac{1}{t_n},
\end{eqnarray*}
$D\in \mathbb{R}$, and $G^{\prime}(t)$ is  the derivative of $G(t)$.
\end{Proposition}

It is known that
Lagrange interpolation functions  often have distinct formats at different sampling times.
Hence it is very complicated to recover signals by using Proposition~\ref{nonuniform:n1}.
In this paper, we give another recovery approach instead.
We begin with a nonuniform sampling model \cite{XZT2016} as in Fig.~\ref{Fig.3}, where $\{x(t_n)\}$ is the sampling sequence of a random signal $x(t)$, and $\{t_n\}$ is the sequence of sampling points.
We assume that  $t_n=nT+\xi_n$, where $T\le \frac{\pi b}{u_r}$ is the average sampling interval,
and  $\{\xi_n\}$ is a sequence of independent identically distributed (i.i.d.) random variables with zero mean in the interval $(-T/2,T/2)$. This nonuniform sampling model is also called jitter sampling.
For jitter sampling, the sampling noise is introduced to the expected sampling time, i.e., $t_n=nT+\tau_n$ with $\tau_n\in(-T/2,T/2)$ and $\mathbb{E}[\tau_n]=0$. For example, the time base jitter of a $50$ GHz sampling oscilloscope is identified to have standard deviation $0.965$ ps, that is, the actual measurement time is corrupted by zero-mean Gaussian noise \cite{verbe2006}. The jittered samples often occur in biomedical devices and A/D converters due to the internal clock imperfections \cite{Amir2005,Janik2004}. 

Next we show that in the sense of second order statistic characters, nonuniform sampling is identical to uniform sampling after a pre-filter.

\begin{Theorem}\label{Thm:LCT:A}
Assume that a random signal $x(t)$ is bandlimited in the LCT domain with bandwidth $u_r$, and $\tilde{x}(t)=x(t)e^{j\frac{a}{2b}t^2}$ is a wide sense stationary process. Then, in the sense of second order statistic characters, the nonuniform sampling of $x(t)$ is identical to the uniform sampling after a LCT filter $h_1(t)$ established in Fig.~\ref{Fig.4}, i.e.,
\begin{equation}\label{Thm:LCT:A1}
h_1(t)=\frac{1}{\sqrt{2\pi}}\int_{-\infty}^{+\infty}H_1(u)e^{jut/b}\textrm{d}u,
\end{equation}
where $T$ is the average sampling interval, $t_n=nT+\xi_n$ is the sampling point, $v(t)$ is an additive noise with zero mean and is independent of $x(t)$, and the LCT auto-power spectral density $P_{vv}^{A}(u)$ of $v(t)$ is $P_{xx}^{A}(u)(1-|H_1(u)|^2).$ Here, $H_1(u)=\phi_\xi(\frac{u}{b})$, where $\phi_\xi(u)$ denotes the characteristic function of $\xi_n$.
\end{Theorem}

\begin{proof}
Since $\tilde{x}(t)=x(t)e^{j\frac{a}{2b}t^2}$ is a wide sense stationary process,
the nonuniform sampling can be described as in Fig~\ref{Fig.5}. By the design of LCT filter in Fig~\ref{Fig.4}, we have
\begin{equation}\label{Thm:LCT:A2}
P_{yy}^A(u)=|H_1(u)|^2P_{xx}^A(u).
\end{equation}
Since $v(t)$ is an additive noise with zero mean
and $v(t)$  is independent of $x(t)$,
the LCT auto-correlation function $R_{zz}^A(\tau)$ of $z(t_n)$ is identical to $R_{yy}^{A}(\tau)$.
Thus, combining (\ref{spectrum:A1}) and (\ref{Thm:LCT:A2}), we have
\begin{eqnarray}
R_{zz}^A(nT,(n-k)T)
&=&R_{yy}^A(nT,(n-k)T)\nonumber\\
&=&\int_{-u_r}^{u_r}P_{yy}^A(u)e^{-j\frac{a}{2b}(kT)^2+j\frac{1}{b}ukT}{\rm{d}}u\nonumber\\
&=&\int_{-u_r}^{u_r}|H_1(u)|^2P_{xx}^A(u)e^{-j\frac{a}{2b}(kT)^2+j\frac{1}{b}ukT}{\rm{d}}u.\label{Thm:LCT:A3}
\end{eqnarray}
Hence
\begin{eqnarray}
R_{\tilde{z}\tilde{z}}(nT,(n-k)T)
&=&\mathbb{E}[z(nT)z^*(nT-kT)e^{j\frac{a}{2b}((nT)^2-(nT-kT)^2)}]\nonumber\\
&=&e^{j\frac{a}{2b}(kT)^2}R_{zz}^{A}(nT,(n-k)T)\nonumber\\
&=&\int_{-u_r}^{u_r}|H_1(u)|^2P_{xx}^A(u)e^{j\frac{1}{b}ukT}{\rm{d}}u.\label{Thm:LCT:A4}
\end{eqnarray}

Since $x(t_n)$ and $\xi_n$ are two random variables, by (\ref{spectrum:A1}), we obtain
\begin{align}
&R_{xx}^A(t_n,t_{n-k})\nonumber\\
=&\mathbb{E}[R_{xx}^A(kT+\xi_n-\xi_{n-k})]\nonumber\\
=&\int_{-u_r}^{u_r}P_{xx}^A(u)\mathbb{E}[e^{-j\frac{a}{2b}(kT+\xi_n-\xi_{n-k})^2+j\frac{1}{b}u(kT+\xi_n-\xi_{n-k})}]{\rm{d}}u.\label{Thm:LCT:A5}
\end{align}
Combining (\ref{LCT:A1}) and (\ref{spectrum:A1}), we have
\begin{align}
&\mathbb{E}[R_{\tilde{x}\tilde{x}}(kT+\xi_n-\xi_{n-k})]\nonumber\\
=&\mathbb{E}[e^{j\frac{a}{2b}(kT+\xi_n-\xi_{n-k})^2}R_{xx}^A(kT+\xi_n-\xi_{n-k})]\nonumber\\
=&\int_{-u_r}^{u_r}P_{xx}^A(u)e^{j\frac{1}{b}ukT}\mathbb{E}[e^{j\frac{1}{b}u(\xi_n-\xi_{n-k})}]{\rm{d}}u.\label{Thm:LCT:A6}
\end{align}
Let $Z=\xi_n-\xi_{n-k}$  and  $f_Z(\eta)$ be the probability density function of $Z$.
Note that $\xi_n$ and $\xi_{n-k}$ are independent and have identical distributions.
Let  $f_{\xi}(\eta)$ be their common probability density function. Then we have
\begin{equation}\label{Thm:LCT:A7}
f_Z(\eta)=f_{\xi}(\eta)\star f_{\xi}(-\eta),
\end{equation}
where $\star$ denote as the convolution operator. Hence
\begin{eqnarray}
\mathbb{E}[e^{j\frac{1}{b}u(\xi_n-\xi_{n-k})}]
&=&\int_{-\infty}^{+\infty}f_Z(\eta)e^{j\frac{1}{b}u\eta}{\rm{d}}\eta\nonumber\\
&=&\int_{-\infty}^{+\infty}[f_{\xi}(\eta)\star f_{\xi}(-\eta)]e^{j\frac{1}{b}u\eta}{\rm{d}}\eta\nonumber\\
&=&\int_{-\infty}^{+\infty}f_{\xi}(\eta)e^{j\frac{1}{b}u\eta}{\rm{d}}\eta
    \times \int_{-\infty}^{+\infty}f_{\xi}(-\eta)e^{j\frac{1}{b}u\eta}{\rm{d}}\eta \nonumber\\
&=&|\phi_{\xi}(\frac{u}{b})|^2,\label{Thm:LCT:A8}
\end{eqnarray}
where
\[
\phi_{\xi}(u)=\int_{-\infty}^{+\infty}f_{\xi}(\eta)e^{ju\eta}{\rm{d}}\eta.
\]
Substituting (\ref{Thm:LCT:A8}) into (\ref{Thm:LCT:A6}), we get
\begin{eqnarray}
\mathbb{E}[R_{\tilde{x}\tilde{x}}(kT+\xi_n-\xi_{n-k})]
&=&\int_{-u_r}^{u_r}P_{xx}^A(u)e^{j\frac{1}{b}ukT}\mathbb{E}[e^{j\frac{1}{b}u(\xi_n-\xi_{n-k})}]{\rm{d}}u\nonumber\\
&=&\int_{-u_r}^{u_r}P_{xx}^A(u)|\phi_{\xi}(\frac{u}{b})|^2e^{j\frac{1}{b}ukT}{\rm{d}}u.\label{Thm:LCT:A9}
\end{eqnarray}
By setting $H_1(u)=\phi_{\xi}(\frac{u}{b})$ in (\ref{Thm:LCT:A4}),  we get  (\ref{Thm:LCT:A9}).
Hence, the auto-correlation function of $\tilde{x}(t_n)$ in Fig.~\ref{Fig.5} is equal to that of the output in Fig.~\ref{Fig.4}. Therefore, the nonuniform sampling is identical to the uniform sampling in Fig.~\ref{Fig.5}, in the sense of second order statistic characters. This completes the proof.
\end{proof}

\subsection{Approximate recovery approach}

As we mention above, it is much easier to deal with uniform sampling than nonuniform sampling. Since sinc interpolation leads to exact recovery for uniform sampling, Maymon and Oppenheim \cite{MO2011} introduced a new approximate recovery formula of nonuniform sampling for a random signal $x(t)$ bandlimited in Fourier domain by utilizing sinc interpolation function. The approximate recovery formula can be represented as follows:
\begin{equation}\label{appxo}
\hat{x}(t)=\frac{T}{T_N}\sum_{n=-\infty}^{+\infty}x(t_n)s(t-\tilde{t}_n),
\end{equation}
where $s(t)={\rm{sinc}}\big(\frac{\pi t}{T_N}\big),$\; $\frac{\pi}{T_N}$ is the bandwidth of $x(t)$,\;$\tilde{t}_n=nT+\zeta_n$. Here, $\tilde{t}_n$ is not required to
be identical to the original sampling points $t_n$. But, if the original random signal $x(t)$ is not bandlimited in the Fourier domain, the approximate recovery approach might
not work. Motivated by \cite{MO2011}, Xu, Zhang, and Tao \cite{XZT2016} considered the case when the random signal is bandlimited in the fractional Fourier domain.
Since LCT is a more general transform, which includes Fourier transform and fractional Fourier transform as its special cases, it is possible that a signal which is non-bandlimited in the Fourier domain or the fractional Fourier domain, is bandlimited in the LCT domain. So, it is necessary to investigate the corresponding approximate recovery result for random signal bandlimited in the LCT domain.

\begin{Theorem}\label{LCT:B1}
Assume that random a signal $x(t)$ is bandlimited in the LCT domain with bandwidth $u_r$, and $\tilde{x}(t)=x(t)e^{j\frac{a}{2b}t^2}$ is a wide sense stationary process.
Then $x(t)$ can be approximated from its nonuniform samples by utilizing the sinc interpolation function,
\begin{equation}\label{LCT:B2}
\hat{x}(t)=\frac{T}{T_N}e^{-j\frac{a}{2b}t^2}\sum_{n=-\infty}^{+\infty}x(t_n)e^{j\frac{a}{2b}t_n^2}h_2(t-\tilde{t}_n),
\end{equation}
where $h_2(t)={\rm{sinc}}\big(\frac{u_rt}{b}\big),$\; $T$ is the uniform sampling interval, $T_N$ is the Nyquist sampling interval, $t_n=nT+\xi_n,$\; $\tilde{t}_n=nT+\zeta_n$, and $\xi_n$ is not necessarily equal to $\zeta_n$.
\end{Theorem}

\begin{proof}
From Lemma~\ref{Lem:L1}, we know that $\tilde{x}(t)$ is bandlimited in the Fourier domain with bandwidth $\frac{u_r}{b}.$
By (\ref{appxo}), we have
\begin{equation}\label{LCT:B3}
\bar{x}(t)=\sum_{n=-\infty}^{+\infty}\frac{T}{T_N}\tilde{x}(t_n){\rm{sinc}}\frac{\pi(t-\tilde{t}_n)}{T_N}.
\end{equation}
Substituting $\tilde{x}(t)=x(t)e^{j\frac{a}{2b}t^2}$ and  $\bar{x}(t)=\hat{x}(t)e^{j\frac{a}{2b}t^2}$ into (\ref{LCT:B3}), we  obtain (\ref{LCT:B2}), which completes the proof.
\end{proof}
From Theorem~\ref{LCT:B1}  we know that the approximate recovery approach for random signals  which are bandlimited in the LCT domain, can be expressed as the model presented in Fig.~\ref{Fig.6}.

\subsection{Error estimate of reconstruction for random signals}

In this subsection, we study the reconstruction error in the mean square sense by considering the sampling and reconstruction procedures as the system whose frequency response is dependent on the probability density function of the perturbations. It follows from Theorem~\ref{Thm:LCT:A} that, if the average sampling interval $T$ is greater than the Nyquist sampling interval $T_N$, then the model presented in Fig.~\ref{Fig.7} is identical to the procedure, which includes the nonuniform sampling mentioned in subsection~\ref{subsec:N1}, and the approximate reconstruction approach by using sinc interpolation function described in Fig.~\ref{Fig.6}, in the sense of second order statistic characters.

\begin{Theorem}\label{Thm:Error}
Assume that a random signal $x(t)$ is bandlimited in the LCT domain with bandwidth $u_r$, $\hat{x}(t)$ is the approximation of $x(t)$ obtained in Fig.~\ref{Fig.7},
and $\tilde{x}(t)=x(t)e^{j\frac{a}{2b}t^2}$ is a wide sense stationary process.
Let the frequency response of the filter $h_3(t)$ be the joint characteristic function of the random variables $\xi_{n}$ and $\zeta_{n}$, i.e., $\phi_{\xi\zeta}(u,-u)$.
And let $v(t)$ be an additive noise with zero mean, which is uncorrelated with $x(t)$ and has the power spectral density
\begin{equation}\label{Error:E1}
P_{vv}(u)=T\int_{-u_r}^{u_r}P_{xx}^A(u_1)[1-\phi_{\xi\zeta}(\frac{u_1}{b},-u)|^2]{\rm{d}}u_1,\quad |u|<u_r.
\end{equation}
Then the model described in Fig.~\ref{Fig.7} is identical to the procedure,
which includes the nonuniform sampling mentioned in subsection~\ref{subsec:N1} and the approximate reconstruction approach by utilizing sinc interpolation function represented in Fig.~\ref{Fig.6}, in the sense of second order statistic characters.
Furthermore, we have
\begin{eqnarray*}\label{mseerror}
\mathbb{E}[|\hat{x}(t)-x(t)|^2]
&=&\int_{-u_r}^{u_r}P_{xx}^{A}(u)|1-\phi_{\xi\zeta}(\frac{u}{b},-\frac{u}{b})|^2{\rm{d}}u{}\nonumber\\
   {}&&\qquad+\frac{T}{2\pi b}\int_{-u_r}^{u_r}P_{xx}^{A}(u)\int_{-u_r}^{u_r}1-|\phi_{\xi\zeta}(\frac{u}{b},-\frac{u_1}{b})|^2{\rm{d}}u_1{\rm{d}}u.
\end{eqnarray*}
\end{Theorem}

\begin{proof}
By Theorem~\ref{LCT:B1}, we have
\begin{equation}\label{Error:E2}
\bar{x}(t)=\hat{x}(t)e^{-j\frac{a}{2b}t^2}=\frac{T}{T_N}\sum_{n=-\infty}^{+\infty}x(t_n)e^{j\frac{a}{2b}t_n^2}h_2(t-\tilde{t}_n),
\end{equation}
where $h_2(t)={\rm{sinc}}\big(\frac{u_r}{b}t\big)$.
Thus, we can get the auto-correlation function of $\bar{x}(t)$ as follows:
\begin{align}
R_{\bar{x}\bar{x}}(t,t-\tau)=&\mathbb{E}[\bar{x}(t)\bar{x}(t-\tau)]\nonumber\\
=&\mathbb{E}\Big[\frac{T}{T_N}\sum_{n=-\infty}^{+\infty}x(t_n)e^{j\frac{a}{2b}t_n^2}h_2(t-\tilde{t}_n)\nonumber\\
    &\quad\times\frac{T}{T_N}\sum_{k=-\infty}^{+\infty}x^*(t_k)e^{-j\frac{a}{2b}t_k^2}h_2^*(t-\tau-\tilde{t}_k)\Big]\nonumber\\
=&\Big(\frac{T}{T_N}\Big)^2\mathbb{E}\Big[\sum_{n=-\infty}^{+\infty}x(nT+\xi_n)e^{j\frac{a}{2b}(nT+\xi_n)^2}h_2(t-nT-\zeta_n){}\nonumber\\
   &\quad\times \sum_{k=-\infty}^{+\infty}x^*(kT+\xi_k)e^{-j\frac{a}{2b}(kT+\xi_k)^2}h_2^*(t-\tau-kT-\zeta_{k})\Big]\nonumber\\
=&\Big(\frac{T}{T_N}\Big)^2\sum_{n=-\infty}^{+\infty}\sum_{k=-\infty}^{+\infty}
        \mathbb{E}\Big[R_{\tilde{x}\tilde{x}}(nT-kT+\xi_n-\xi_k){}\nonumber\\
   &\quad\times h_2(t-nT-\zeta_n)h_2^*(t-\tau-kT-\zeta_{k})\Big]\nonumber\\
=&\Big(\frac{T}{T_N}\Big)^2R_{\tilde{x}\tilde{x}}(0)\sum_{n=-\infty}^{+\infty}\mathbb{E}\Big[h_2(t-nT-\zeta_n)h_2^*(t-\tau-nT-\zeta_{n})\Big]\nonumber\\
   &\quad+\Big(\frac{T}{T_N}\Big)^2\sum_{n\ne k}\mathbb{E}\Big[R_{\tilde{x}\tilde{x}}(nT-kT+\xi_n-\xi_k)h_2(t-nT-\zeta_n){}\nonumber\\
   &\qquad\times h_2^*(t-\tau-kT-\zeta_{k})\Big]\nonumber\\
\triangleq &I+II.\label{Error:E3}
\end{align}
Next, we use two steps to compute $R_{\bar{x}\bar{x}}(t,t-\tau)$.
Note that
\begin{equation}\label{Error:E5}
\sum_{n}e^{j(u_2-u_1)nT}=2\pi\sum_{k}\delta((u_2-u_1)T-2\pi k)
\end{equation}
and
\begin{eqnarray}
h_2(t)&=&\frac{1}{\sqrt{2\pi}}\int_{-\infty}^{+\infty}H_2(u)e^{j\frac{u}{b}t}\textrm{d}u\nonumber\\
&=&\frac{b}{\sqrt{2\pi}}\int_{-\infty}^{+\infty}H_2(ub)e^{jut}\textrm{d}u.\label{Error:E6}
\end{eqnarray}
We obtain
\begin{align}
I=&\Big(\frac{T}{T_N}\Big)^2R_{\tilde{x}\tilde{x}}(0)\sum_{n=-\infty}^{+\infty}\mathbb{E}\Big[h_2(t-nT-\zeta_n)h_2^*(t-\tau-nT-\zeta_{n})\Big]\nonumber\\
=&\frac{b^2}{2\pi}\Big(\frac{T}{T_N}\Big)^2R_{\tilde{x}\tilde{x}}(0)\int_{-\infty}^{+\infty}\int_{-\infty}^{+\infty}H_2(bu_1)H_2^*(bu_2)
   e^{j(u_1-u_2)t}e^{ju_2\tau}{}\nonumber\\
   &\quad\times\sum_{n=-\infty}^{+\infty}e^{j(u_2-u_1)nT}\mathbb{E}[e^{j(u_2-u_1)\zeta_n}]{\rm{d}}u_1{\rm{d}}u_2\nonumber\\
=&\Big(\frac{Tb}{T_N}\Big)^2R_{\tilde{x}\tilde{x}}(0)\int_{-\frac{u_r}{b}}^{\frac{u_r}{b}}\frac{1}{T}|H_2(bu)|^2e^{ju\tau}{\rm{d}}u\nonumber\\
=&\frac{T}{4\pi^2}\int_{-\frac{u_r}{b}}^{\frac{u_r}{b}}e^{ju\tau}\Big[\int_{-\frac{u_r}{b}}^{\frac{u_r}{b}}
    P_{\tilde{x}\tilde{x}}(u_1){\rm{d}}u_1\Big]{\rm{d}}u,\label{Error:E7}
\end{align}
where we use the fact that $H_2(bu)=\frac{T_N}{b\sqrt{2\pi}}\chi_{[-\frac{\pi}{T_N},\frac{\pi}{T_N}]}(u)$ in the last step.

Similarly, we have
\begin{align}
II=&\Big(\frac{T}{T_N}\Big)^2\sum_{n\ne k}\mathbb{E}\Big[R_{\tilde{x}\tilde{x}}(nT-kT+\xi_n-\xi_k)h_2(t-nT-\zeta_n){}\nonumber\\
   &\times h_2^*(t-\tau-kT-\zeta_{k})\Big] \nonumber\\
=&\frac{b^2}{(2\pi)^2}\Big(\frac{T}{T_N}\Big)^2\sum_{n\ne k}\mathbb{E}\Big[\int_{-\frac{u_r}{b}}^{\frac{u_r}{b}}
       P_{\tilde{x}\tilde{x}}(u)e^{ju(nT-kT+\xi_n-\xi_k)}{\rm{d}}u\nonumber\\
   &\times\int_{-\infty}^{+\infty}H_2(bu_1)e^{ju_1(t-nT-\zeta_n)}{\rm{d}}u_1\int_{-\infty}^{+\infty}H_2^*(bu_2)
       e^{-ju_2(t-\tau-kT-\zeta_{k})}{\rm{d}}u_2\Big]\nonumber\\
=&\frac{b^2}{(2\pi)^2}\Big(\frac{T}{T_N}\Big)^2\sum_{n\ne k}\int_{-\frac{u_r}{b}}^{\frac{u_r}{b}}
       \int_{-\infty}^{+\infty}\int_{-\infty}^{+\infty}P_{\tilde{x}\tilde{x}}(u)H_2(bu_1)H_2^*(bu_2) e^{ju_2\tau}\nonumber\\
   &\times e^{j(u_1-u_2)t}e^{j(u-u_1)nT}e^{-j(u-u_2)kT}\mathbb{E}\Big[e^{ju\xi_n}e^{-ju\xi_k}\nonumber\\
      &\quad\times e^{-ju_1\zeta_n}e^{ju_2\zeta_k}\Big]{\rm{d}}u_1{\rm{d}}u_2{\rm{d}}u\nonumber\\
=&\frac{b^2}{(2\pi)^2}\Big(\frac{T}{T_N}\Big)^2\int_{-\frac{u_r}{b}}^{\frac{u_r}{b}}\int_{-\infty}^{+\infty}\int_{-\infty}^{+\infty}
       P_{\tilde{x}\tilde{x}}(u)H_2(bu_1)H_2^*(bu_2)\phi_{\xi\zeta}(u,-u_1)\nonumber\\
     &\times\phi_{\xi\zeta}^*(u,-u_2)e^{ju_2\tau}e^{j(u_1-u_2)t}\sum_{n}e^{j(u-u_1)nT}\nonumber\\
       &\quad\times\sum_{k}e^{-j(u-u_2)kT}{\rm{d}}u_1{\rm{d}}u_2{\rm{d}}u-\frac{b^2}{(2\pi)^2}\Big(\frac{T}{T_N}\Big)^2
        \int_{-\frac{u_r}{b}}^{\frac{u_r}{b}}\int_{-\infty}^{+\infty}\int_{-\infty}^{+\infty}\nonumber\\
         &\qquad\times P_{\tilde{x}\tilde{x}}(u)H_2(bu_1)H_2^*(bu_2)\phi_{\xi\zeta}(u,-u_1)\phi_{\xi\zeta}^*(u,-u_2)
               e^{ju_2\tau}e^{j(u_1-u_2)t}\nonumber\\
            &\quad\qquad\times\sum_{n}e^{j(u_2-u_1)nT}{\rm{d}}u_1{\rm{d}}u_2{\rm{d}}u\nonumber\\
=&\Big(\frac{b}{T_N}\Big)^2\int_{-\frac{u_r}{b}}^{\frac{u_r}{b}}
       P_{\tilde{x}\tilde{x}}(u)|H_2(bu)|^2|\phi_{\xi\zeta}(u,-u)|^2e^{ju\tau}{\rm{d}}u\nonumber\\
   &-\frac{T}{2\pi}\Big(\frac{b}{T_N}\Big)^2\int_{-\frac{u_r}{b}}^{\frac{u_r}{b}}\int_{-\frac{u_r}{b}}^{\frac{u_r}{b}}
       P_{\tilde{x}\tilde{x}}(u_1)|H_2(bu)|^2|\phi_{\xi\zeta}(u_1,-u)|^2e^{ju\tau}{\rm{d}}u_1{\rm{d}}u\nonumber\\
=&\Big(\frac{b}{T_N}\Big)^2\int_{-\frac{u_r}{b}}^{\frac{u_r}{b}}|H_2(bu)|^2e^{ju\tau}
       \Big[P_{\tilde{x}\tilde{x}}(u)|\phi_{\xi\zeta}(u,-u)|^2-\frac{T}{2\pi}\int_{-\frac{u_r}{b}}^{\frac{u_r}{b}}P_{\tilde{x}\tilde{x}}(u_1)\nonumber\\
    &\times|\phi_{\xi\zeta}(u_1,-u)|^2{\rm{d}}u_1\Big]{\rm{d}}u\nonumber\\
=&\frac{1}{2\pi}\int_{-\frac{u_r}{b}}^{\frac{u_r}{b}}e^{ju\tau}
       \Big[P_{\tilde{x}\tilde{x}}(u)|\phi_{\xi\zeta}(u,-u)|^2-\frac{T}{2\pi}\int_{-\frac{u_r}{b}}^{\frac{u_r}{b}}P_{\tilde{x}\tilde{x}}(u_1)\nonumber\\
    &\times|\phi_{\xi\zeta}(u_1,-u)|^2{\rm{d}}u_1\Big]{\rm{d}}u.\label{Error:E8}
\end{align}
Combining (\ref{Error:E7}) and (\ref{Error:E8}), we obtain
\begin{align}
R_{\bar{x}\bar{x}}(t,t-\tau)
=&\frac{T}{4\pi^2}\int_{-\frac{u_r}{b}}^{\frac{u_r}{b}}\Big(\int_{-\frac{u_r}{b}}^{\frac{u_r}{b}}
    P_{\tilde{x}\tilde{x}}(u_1)[1-|\phi_{\xi\zeta}(u_1,-u)|^2]{\rm{d}}u_1\Big)e^{ju\tau}{\rm{d}}u\nonumber\\
  &\quad+\frac{1}{2\pi}\int_{-\frac{u_r}{b}}^{\frac{u_r}{b}}e^{ju\tau}P_{\tilde{x}\tilde{x}}(u)|\phi_{\xi\zeta}(u,-u)|^2{\rm{d}}u.\label{Error:E9}
\end{align}

Similarly, we can obtain the cross-correlation function of $\bar{x}(t)$ and $\tilde{x}(t)$ as follows,
\begin{equation}\label{Error:E10}
R_{\bar{x}\tilde{x}}(t,t-\tau)=\frac{1}{2\pi}\int_{-\frac{u_r}{b}}^{\frac{u_r}{b}}e^{ju\tau}P_{\tilde{x}\tilde{x}}(u)\phi_{\xi\zeta}(u,-u){\rm{d}}u.
\end{equation}
Therefore, we have
\begin{eqnarray}
P_{\bar{x}\bar{x}}(u)&=&P_{\tilde{x}\tilde{x}}(u)|\phi_{\xi\zeta}(u,-u)|^2{}\nonumber\\
   {}&&+\frac{T}{2\pi}\Big(\int_{-\frac{u_r}{b}}^{\frac{u_r}{b}}
        P_{\tilde{x}\tilde{x}}(u_1)[1-|\phi_{\xi\zeta}(u_1,-u)|^2]{\rm{d}}u_1\Big) \label{Error:E11}
\end{eqnarray}
and
\begin{equation}\label{Error:E12}
P_{\bar{x}\tilde{x}}(u)=P_{\tilde{x}\tilde{x}}(u)\phi_{\xi\zeta}(u,-u).
\end{equation}
Hence, the first term $P_{\tilde{x}\tilde{x}}(u)|\phi_{\xi\zeta}(u,-u)|^2$ in (\ref{Error:E11}) is the power spectral density of $\tilde{y}(t)$ in Fig.~\ref{Fig.7}.
Substituting (\ref{Lem:L2}) into (\ref{Error:E1}), we have
\begin{eqnarray*}
P_{vv}(u)&=&T\int_{-u_r}^{u_r}P_{xx}^A(u_1)[1-\phi_{\xi\zeta}(\frac{u_1}{b},-u)|^2]{\rm{d}}u_1\\
&=&T\int_{-u_r}^{u_r}\frac{1}{2\pi b}P_{\tilde{x}\tilde{x}}(\frac{u_1}{b})[1-\phi_{\xi\zeta}(\frac{u_1}{b},-u)|^2]{\rm{d}}u_1\\
&=&\frac{T}{2\pi}\int_{-\frac{u_r}{b}}^{\frac{u_r}{b}}P_{\tilde{x}\tilde{x}}(u_1)[1-\phi_{\xi\zeta}(u_1,-u)|^2]{\rm{d}}u_1.
\end{eqnarray*}
Thus, the second term $\frac{T}{2\pi}\Big(\int_{-\frac{u_r}{b}}^{\frac{u_r}{b}}P_{\tilde{x}\tilde{x}}(u_1)[1-|\phi_{\xi\zeta}(u_1,-u)|^2]{\rm{d}}u_1\Big)$ in (\ref{Error:E11}) is identical to the power spectral density of $v(t)$. Consequently, the model described in Fig.~\ref{Fig.7} is equal to the procedure, which includes the nonuniform sampling mentioned in subsection~\ref{subsec:N1} and the approximate reconstruction approach by utilizing sinc interpolation function represented in Fig.~\ref{Fig.6}, in the sense of second order statistic characters.

Next, we estimate the error $\mathbb{E}[|\hat{x}(t)-x(t)|^2]$.
Let $\epsilon(t)=\hat{x}(t)-x(t)$. Combining (\ref{Lem:L2}) and (\ref{Error:E11}), we get
\begin{eqnarray}
P_{\hat{x}\hat{x}}^{A}(u)
&=&\frac{1}{2\pi b}P_{\bar{x}\bar{x}}(\frac{u}{b})\nonumber\\
&=&\frac{1}{2\pi b}P_{\tilde{x}\tilde{x}}(\frac{u}{b})|\phi_{\xi\zeta}(\frac{u}{b},-\frac{u}{b})|^2{}\nonumber\\
   {}&&+\frac{T}{4\pi^2b}\Big(\int_{-\frac{u_r}{b}}^{\frac{u_r}{b}}
        P_{\tilde{x}\tilde{x}}(u_1)[1-|\phi_{\xi\zeta}(u_1,-\frac{u}{b})|^2]{\rm{d}}u_1\Big)\nonumber\\
&=&P_{xx}^{A}(u)|\phi_{\xi\zeta}(\frac{u}{b},-\frac{u}{b})|^2{}\nonumber\\
   {}&&+\frac{T}{2\pi b}\Big(\int_{-u_r}^{u_r}
        P_{xx}^{A}(u_1)[1-|\phi_{\xi\zeta}(\frac{u_1}{b},-\frac{u}{b})|^2]{\rm{d}}u_1\Big).\label{Error:E13}
\end{eqnarray}
Similarly, we obtain
\begin{eqnarray}
P_{\hat{x}x}^{A}(u)
&=&\frac{1}{2\pi b}P_{\bar{x}\tilde{x}}(\frac{u}{b})\nonumber\\
&=&\frac{1}{2\pi b}P_{\tilde{x}\tilde{x}}(\frac{u}{b})\phi_{\xi\zeta}(\frac{u}{b},-\frac{u}{b}))\nonumber\\
&=&P_{xx}^{A}(u)\phi_{\xi\zeta}(\frac{u}{b},-\frac{u}{b}).\label{Error:E14}
\end{eqnarray}
Hence, the LCT auto-power spectral density of the reconstruction error $\epsilon(t)$ in Fig.~\ref{Fig.7} is
\begin{eqnarray}
P_{\epsilon\epsilon}^A(u)
&=&P_{\hat{x}\hat{x}}^A(u)-P_{\hat{x}x}^A(u)-P_{x\hat{x}}^A(u)+P_{xx}^A(u)\nonumber\\
&=&P_{xx}^{A}(u)|\phi_{\xi\zeta}(\frac{u}{b},-\frac{u}{b})|^2+\frac{T}{2\pi b}\Big(\int_{-u_r}^{u_r}P_{xx}^{A}(u_1){}\nonumber\\
   {}&&\quad\times[1-|\phi_{\xi\zeta}(\frac{u_1}{b},-\frac{u}{b})|^2]{\rm{d}}u_1\Big)-P_{xx}^{A}(u)\phi_{\xi\zeta}(\frac{u}{b},-\frac{u}{b}){}\nonumber\\
   {}&&\qquad-[P_{xx}^{A}(u)\phi_{\xi\zeta}(\frac{u}{b},-\frac{u}{b})]^*+P_{xx}^A(u)\nonumber\\
&=&P_{xx}^{A}(u)|1-\phi_{\xi\zeta}(\frac{u}{b},-\frac{u}{b})|^2+\frac{T}{2\pi b}\Big(\int_{-u_r}^{u_r}P_{xx}^{A}(u_1){}\nonumber\\
   {}&&\qquad\times[1-|\phi_{\xi\zeta}(\frac{u_1}{b},-\frac{u}{b})|^2]{\rm{d}}u_1\Big).\label{Error:E15}
\end{eqnarray}
Therefore, it follows from (\ref{LCT:A1}) and (\ref{spectrum:A1}) that
\begin{eqnarray*}
\mathbb{E}[|\epsilon(t)|^2]
&=&R_{\tilde{\epsilon}\tilde{\epsilon}}(0)\nonumber\\
&=&R_{\epsilon\epsilon}^A(0)\nonumber\\
&=&\int_{-u_r}^{u_r}P_{\epsilon\epsilon}^{A}(u){\rm{d}}u\nonumber\\
&=&\int_{-u_r}^{u_r}P_{xx}^{A}(u)|1-\phi_{\xi\zeta}(\frac{u}{b},-\frac{u}{b})|^2{\rm{d}}u{}\nonumber\\
   {}&&\qquad+\frac{T}{2\pi b}\int_{-u_r}^{u_r}P_{xx}^{A}(u)\int_{-u_r}^{u_r}1-|\phi_{\xi\zeta}(\frac{u}{b},-\frac{u_1}{b})|^2{\rm{d}}u_1{\rm{d}}u.
\end{eqnarray*}
This completes the proof.
\end{proof}

Note that the reconstruction error $\mathbb{E}[|\hat{x}(t)-x(t)|^2]$ is related to the LCT auto-correlation power spectral density $P_{xx}^{A}(u)$ of the random signal $x(t)$ and the joint characteristic function $\phi_{\xi\zeta}(u,-u)$ of two random variables $\xi_n$ and $\zeta_n$. In particular, when $\xi_n$ and $\zeta_n$ are constants and both are equal to zeros, i.e., the nonuniform sampling studied in this paper reduces to uniform sampling, we have $\mathbb{E}[|\hat{x}(t)-x(t)|^2]=0$ from Theorem \ref{Thm:Error}. Therefore the result of uniform sampling proposed in \cite[Theorem 3.4]{HS2015} is a special case of Theorems \ref{LCT:B1} and \ref{Thm:Error} in this paper.

On the other hand, the LCT includes many widely used linear transforms as special cases. For example, by setting $A=
\left(
\begin{array}{cc}
\cos \theta  & \sin \theta \\
-\sin \theta & \cos \theta \\
\end{array}
\right)$ with $\theta\in[-\pi,\pi)$ in (\ref{def:LCT:L0}), the LCT of $f(t)$ becomes the fractional Fourier transform of $f(t)$ with angle $\theta$.
In this case, our results coincide with those in \cite{XZT2016}, and in particular, when $\xi_n$ and $\zeta_n$ are constants and equal to zeros, our results coincide with the uniform sampling results in \cite{TZW2011}. Furthermore, by setting $A=
\left(
\begin{array}{cc}
0  & 1 \\
-1 & 0 \\
\end{array}
\right)$, the LCT of $f(t)$ becomes the Fourier transform of $f(t)$ multiplied by a constant $\sqrt{\frac{1}{j2\pi}}$. In this case, our results coincide with those in \cite{MO2011}.

\subsection{Simulations}
In this subsection, we consider a random signal $x(t)=e^{j5\pi t+j\psi-j\frac{3}{2}\pi t^2}, -4\leq t\leq4$, where $\psi$ follows the standard  Gaussian distribution.
It is easy to verify that $\tilde{x}(t)=x(t)e^{j\frac{a}{2b}t^2}$ is wide sense stationary, and $x(t)$ is approximately bandlimited with bandwidth $10$ Hz in the LCT domain with parameter $A=\left(
\begin{array}{cc}
a & b \\
c & d \\
\end{array}
\right)=\left(
\begin{array}{cc}
3 & 1/\pi \\
\pi & 2/3 \\
\end{array}
\right)$. Let $T=T_N=0.1$. First, the approximate signal recovery results based on (\ref{LCT:B2}) for one realization of $x(t)$ are respectively shown in Fig.~\ref{fig11} in terms of two different nonuniform sampling models. Let $\xi_n$ follow the uniform distribution in the interval $[-0.01-0.002*k,0.01+0.002*k]$ and integer $k$ take a value from $0$ to $15$. Then, for each $k$, $0\leq k\leq15$, we implement $5000$ realizations of $x(t)$ and estimate the mean square error of the reconstruction in terms of four different nonuniform sampling models as shown in Fig.~\ref{fig13}. One can see from Fig.~\ref{fig13} that the reconstruction from the nonuniform sampling with $\zeta_n=0$ is preferable. In fact, the reconstruction is an approximate solution, and we cannot claim which sampling model is the best in general. However, according to (\ref{mseerror}) in Theorem \ref{Thm:Error}, a lower mean square error of the reconstruction might be obtained by choosing a proper joint characteristic function of $\xi_n$ and $\zeta_n$.


\section{Conclusion}\label{sec C}
In this paper, we mainly discuss the nonuniform sampling for random signals, which are bandlimited in the LCT domain. At the beginning, based on the concepts of the LCT correlation function and power spectral density, we get the connection between the LCT auto-power spectral density of the inputs and outputs. Moreover, we show that nonuniform sampling for random signals bandlimited in the LCT domain can be identical to uniform sampling after a pre-filter in the sense of second order statistic characters. Furthermore, we derive an approximate reconstruction formula for random signals bandlimited in LCT domain from their nonuniform samples, by utilizing the sinc interpolation functions. Finally, we investigate the error between the original random signal and its approximation in the mean square sense, and verify the performances of our theoretical results by numerical simulation.



\newpage


\begin{figure}[!ht]
\begin{center}

\begin{tikzpicture}[scale=1.2]
\node (xt) {$x(t)$};
\node(LCT)  [rectangle,draw,right=  of xt]  {LCT};
\node(times)  [right=of LCT]  {$\displaystyle\bigotimes$};
\node(Hu)  [above=of times]  {$H(u)$};
\node(iLCT)  [rectangle,draw,right=of times]  {Inverse LCT};
\node (yt) [right=of iLCT]{$y(t)$};

\draw [->] (xt) to (LCT);
\draw [->] (LCT) to node [label=above:$X(u)$]{} (times);
\draw [->] (times) to node [label=above:$Y(u)$]{} (iLCT);
\draw [->] (iLCT) to (yt);
\draw [->] (Hu) to (times);
\end{tikzpicture}
\end{center}
\caption{The LCT multiplicative filter.}
\label{Fig.1}
\end{figure}
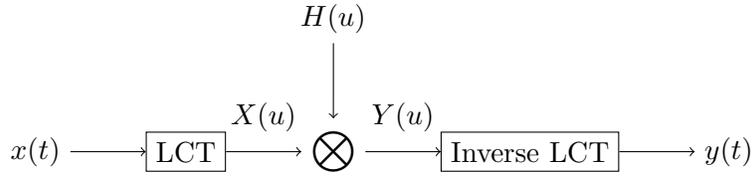

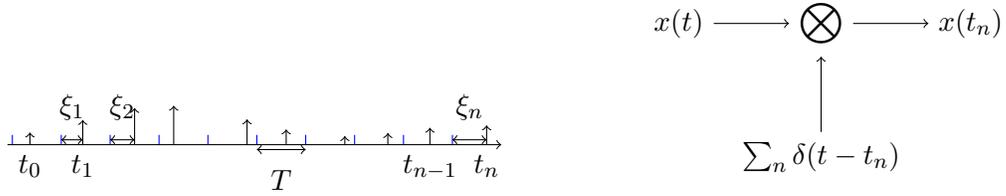
\begin{figure}[!ht]
\begin{minipage}[t]{0.5\linewidth}
\vspace{-1.3cm}
\begin{center}
\begin{tikzpicture}[domain=0:5,scale=1.3]

\foreach \x in {0,0.5,...,4.5}
  \draw[color=blue](\x,0)--(\x,0.1);

\foreach \x in {0.18,0.72,1.25,1.65,2.4,2.8,3.4,3.84,4.27,4.85}
 {\pgfmathparse{0.2+0.1*sin(\x r)-0.1*cos(2*\x r)}\edef\y{\pgfmathresult}
  \draw [->](\x,0) -- (\x,\y);
 }
\fill[black!100] (0.18,0) node[below]{$t_0$};
\fill[black!100] (0.72,0) node[below]{$t_1$};
\fill[black!100] (4.27,0) node[below]{$t_{n-1}$};
\fill[black!100] (4.85,0) node[below]{$t_{n}$};

\draw[->](-0.05,0) -- (5,0);

\foreach \x/\y in {2.5/-0.05}
  \draw [<->](\x,\y) -- node [label=below:$T$]{}+(0.5,0);
\foreach \x/\y in {0.5/0.05}
  \draw [<->](\x,\y) -- node [label=above:$\xi_1$]{}+(0.22,0);
\foreach \x/\y in {1.0/0.05}
  \draw [<->](\x,\y) -- node [label=above:$\xi_2$]{}+(0.25,0);
\foreach \x/\y in {4.5/0.05}
  \draw [<->](\x,\y) -- node [label=above:$\xi_n$]{}+(0.35,0);
\end{tikzpicture}

\end{center}
\end{minipage}
\begin{minipage}[t]{0.5\linewidth}
\begin{center}
\begin{tikzpicture}[scale=1.2]
\node (xt) {$x(t)$};
\node(times)  [right=of xt]  {$\displaystyle\bigotimes$};
\node(delta)  [below=of times]  {$\sum_{n}\delta(t-t_n)$};
\node (xn) [right=of times]{$x(t_n)$};

\draw[->] (xt) to (times);
\draw[->] (delta) to (times);
\draw[->] (times) to (xn);

\end{tikzpicture}

\end{center}


\end{minipage}
\caption{The nonuniform sampling representation.}
\label{Fig.3}
\end{figure}

\begin{figure}[!ht]
\begin{center}

\begin{tikzpicture}[scale=1.2]
\node (xt) {$x(t)$};
\node(times)  [right= of xt]  {$\displaystyle\bigotimes$};
\node(et)  [above=of times]  {$e^{j\frac{a}{2b}t^2}$};
\node(ht)  [rectangle,draw,right=of times]  {$h_1(t)$};
\node(times2)  [right=of ht]  {$\displaystyle\bigotimes$};
\node(e2)  [above=of times2]  {$e^{-j\frac{a}{2b}t^2}$};
\node(AD)  [rectangle,draw,right=of times2]  {Sampling};
\node(add)  [right=of AD]  {$\displaystyle\bigoplus$};
\node (vn) [below=of add]{$v(nT)$};
\node (times3) [right=of add]{$\displaystyle\bigotimes$};
\node(e3)  [above=of times3]  {$e^{j\frac{a}{2b}(nT)^2}$};
\node (zn) [right=of times3]{$\tilde{z}(nT)$};

\draw [->] (xt) to (times);
\draw [->] (et) to (times);
\draw [->](times) to node [label=above:$\tilde{x}(t)$]{} (ht);
\draw [->](ht) to node [label=above:$\tilde{y}(t)$]{} (times2);
\draw [->](e2) to (times2);
\draw [->](times2) to node [label=above:$y(t)$]{} (AD);
\draw [->](AD) to node [label=above:$y(nT)$]{} (add);
\draw [->](add) to node [label=above:$z(nT)$]{} (times3);
\draw [->](vn) to (add);
\draw [->](e3) to (times3);
\draw [->](times3) to (zn);

\end{tikzpicture}
\end{center}
\caption{The equivalent system of the nonuniform sampling, where $v(t)$ is an additive noise with zero mean, $v(t)$ is independent of $x(t)$, and the LCT auto-power spectral density of $v(t)$ is $P_{xx}^{A}(u)(1-|H_1(u)|^2)$.}
\label{Fig.4}
\end{figure}
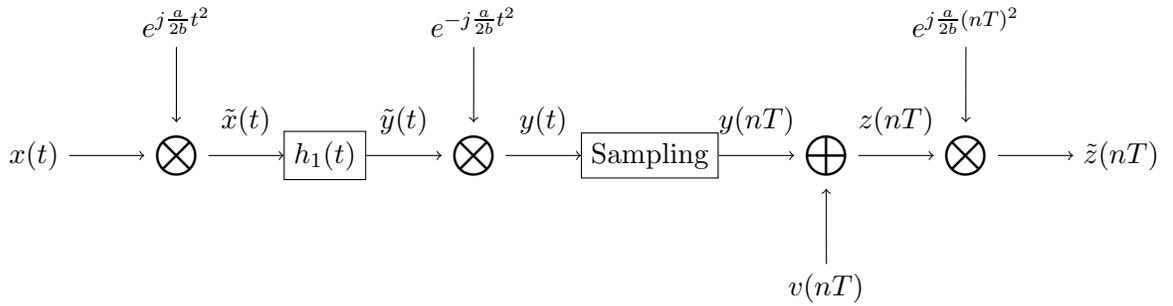

\begin{figure}[!ht]
\begin{center}

\begin{tikzpicture}[scale=2]
\node (xt) {$x(t)$};
\node(sampling)  [rectangle,draw,right=of xt]  {Sampling};
\node(tn)  [below=of sampling]  {$t_n=nT+\xi_n$};
\node(times)  [right=of sampling]  {$\displaystyle\bigotimes$};
\node(et)  [above=of times]  {$e^{j\frac{a}{2b}t_n^2}$};
\node (yt) [right=of times]{$\tilde{x}(t_n)$};

\draw [->] (xt) to (sampling);
\draw [->](tn) to (sampling);
\draw [->](sampling) to node [label=above:$x(t_n)$]{} (times);
\draw [->](et) to (times);
\draw [->](times) to (yt);

\end{tikzpicture}

\end{center}
\caption{Another version of nonuniform sampling.}
\label{Fig.5}
\end{figure}
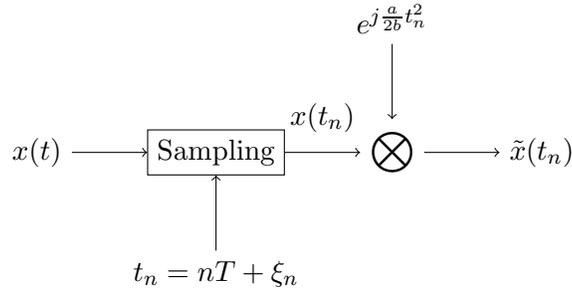

\begin{figure}[!ht]
\begin{center}

\begin{tikzpicture}[scale=2]
\node (xn) {$x(t_n)$};
\node(times)  [right=of xn]  {$\displaystyle\bigotimes$};
\node(et)  [above=of times]  {$e^{j\frac{a}{2b}t_n^2}$};
\node(impulse)  [rectangle,draw,right=of times]  {Synthesis};
\node (times2) [right=of impulse]{$\displaystyle\bigotimes$};
\node(e2)  [above=of times2]  {$e^{-j\frac{a}{2b}t^2}$};
\node (yt) [right=of times2]{$\hat{x}(t)$};
\draw [->](xn) to (times);
\draw [->](et) to (times);
\draw [->](times) to node [label=above:$\tilde{x}(t_n)$]{} (impulse);
\draw [->](impulse) to node [label=above:$\bar{x}(t)$] {} (times2);
\draw [->](e2) to (times2);
\draw [->](times2) to (yt);
\end{tikzpicture}

\end{center}
\caption{The approximate reconstruction with sinc interpolation function, where $\bar x(t)=\sum\limits_n \frac{T}{T_N}\tilde{x}(t_n){\mathrm{sinc}}\frac{\pi(t-\tilde{t}_n)}{T_N}$.}
\label{Fig.6}
\end{figure}
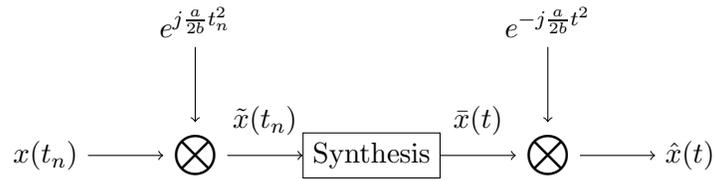

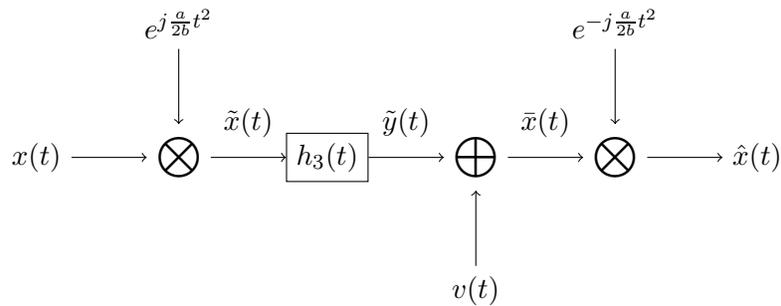
\begin{figure}[!ht]
\begin{center}

\begin{tikzpicture}[scale=2]
\node (xt) {$x(t)$};
\node(times)  [right=of xt]  {$\displaystyle\bigotimes$};
\node(et)  [above=of times]  {$e^{j\frac{a}{2b}t^2}$};
\node(ht)  [rectangle,draw,right=of times]  {$h_3(t)$};
\node(add)  [right=of ht]  {$\displaystyle\bigoplus$};
\node (vt) [below=of add]{$v(t)$};
\node (times2) [right=of add]{$\displaystyle\bigotimes$};
\node(e2)  [above=of times2]  {$e^{-j\frac{a}{2b}t^2}$};
\node (yt) [right=of times2]{$\hat{x}(t)$};

\draw [->](xt) to (times);
\draw [->](et) to (times);
\draw [->](times) to node [label=above:$\tilde{x}(t)$]{} (ht);
\draw [->](ht) to node [label=above:$\tilde{y}(t)$]{} (add);
\draw [->](vt) to (add);
\draw [->](add) to node [label=above:$\bar{x}(t)$]{} (times2);
\draw [->](e2) to (times2);
\draw [->](times2) to (yt);

\end{tikzpicture}

\end{center}
\caption{The nonuniform sampling and reconstruction system,
where $v(t)$ is an additive noise with zero mean, which is uncorrelated with $x(t)$ and has  the power spectral density
defined by (\ref{Error:E1}).}
\label{Fig.7}
\end{figure}

\begin{figure}
\centering
\subfigure[]{
\begin{minipage}[b]{1\textwidth}
\includegraphics[width=1\textwidth]{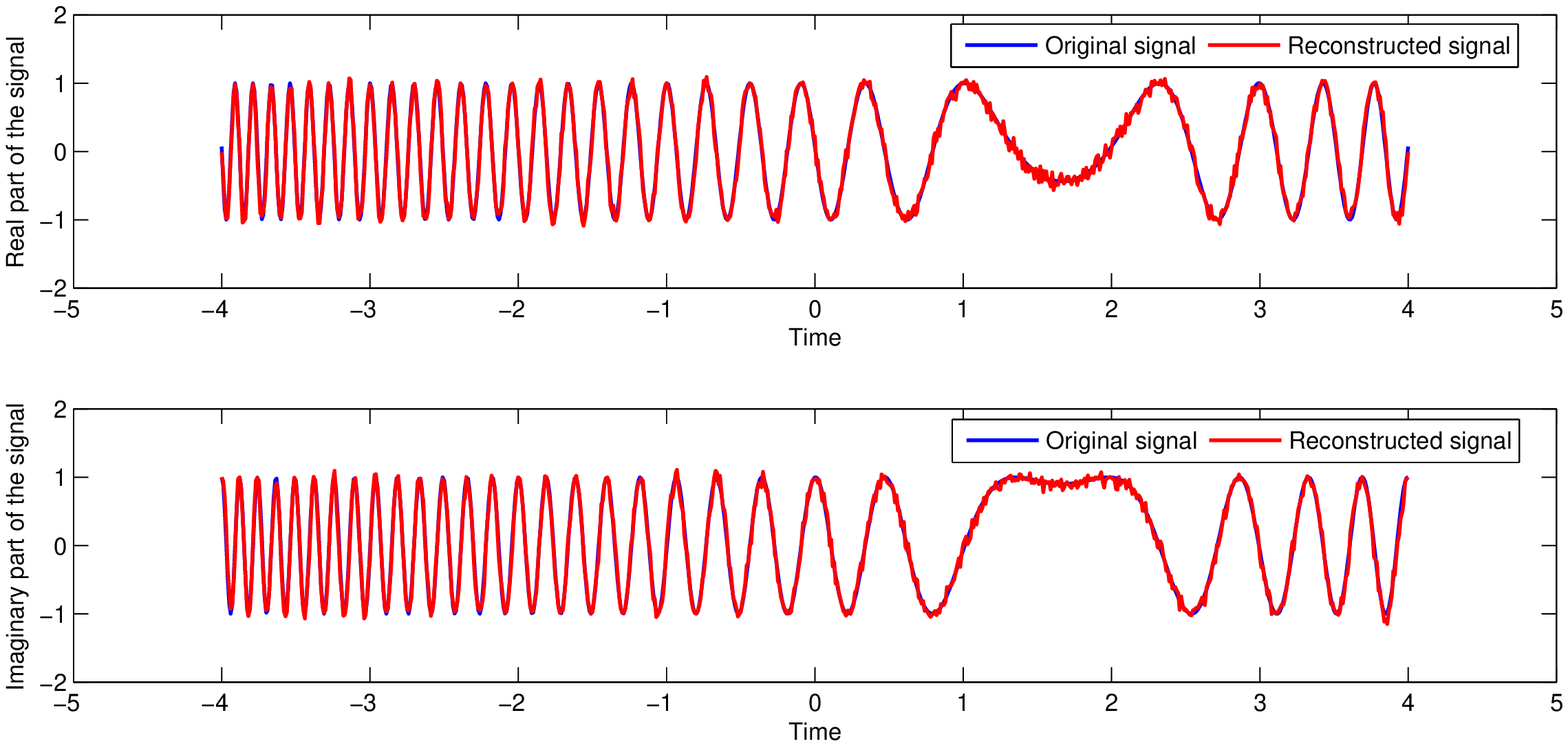}
\end{minipage}
}
\subfigure[]{
\begin{minipage}[b]{1\textwidth}
\includegraphics[width=1\textwidth]{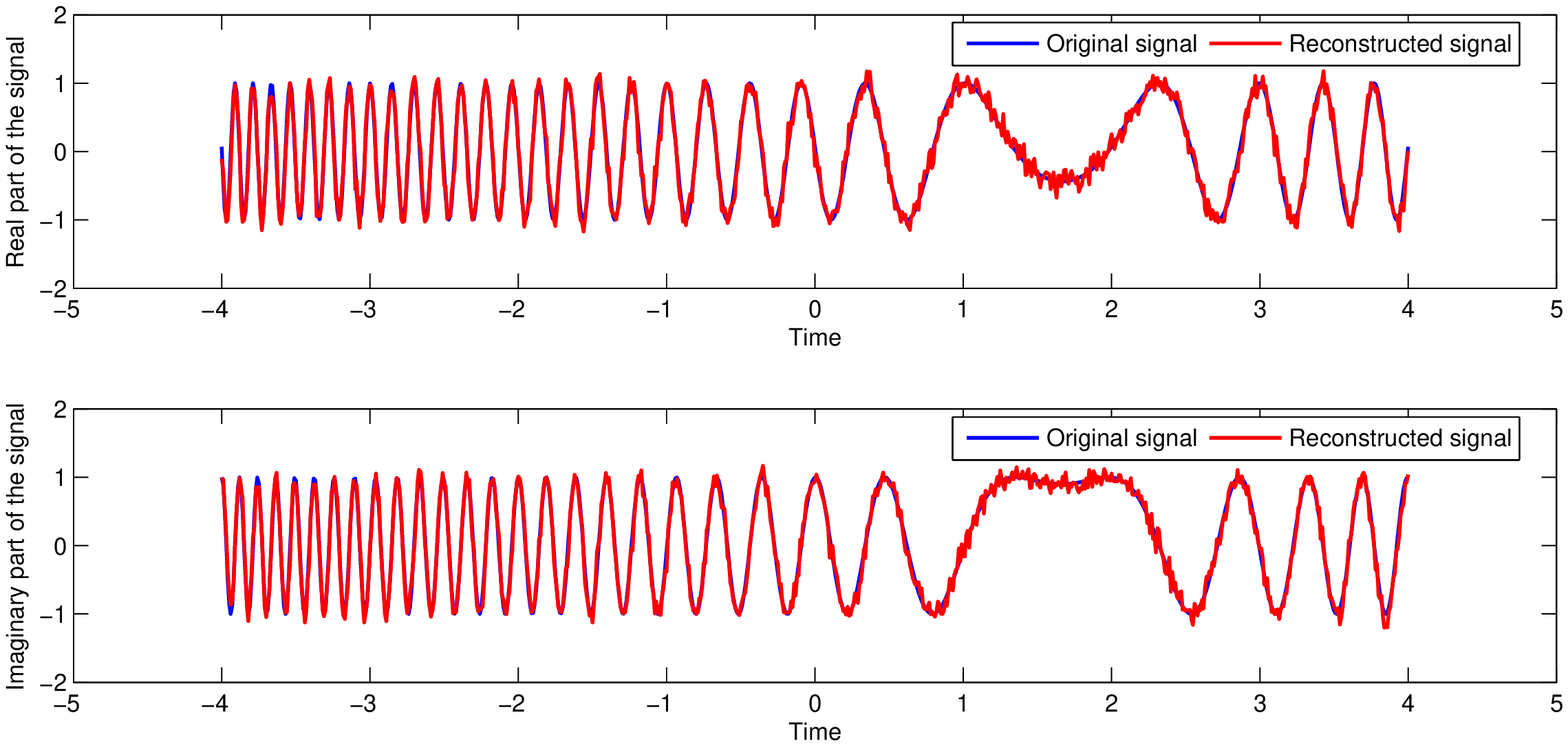}
\end{minipage}
}
 \caption{The approximate signal reconstruction: (a) when $\xi_n$ is uniformly distributed in the interval $[-0.01,0.01]$ and $\zeta_n=0$; (b) when $\xi_n$ and $\zeta_n$ are i.i.d. with uniform distribution in the interval $[-0.01,0.01]$.}
 \label{fig11}
\end{figure}

\begin{figure}[!ht]
\includegraphics[width=1\columnwidth]{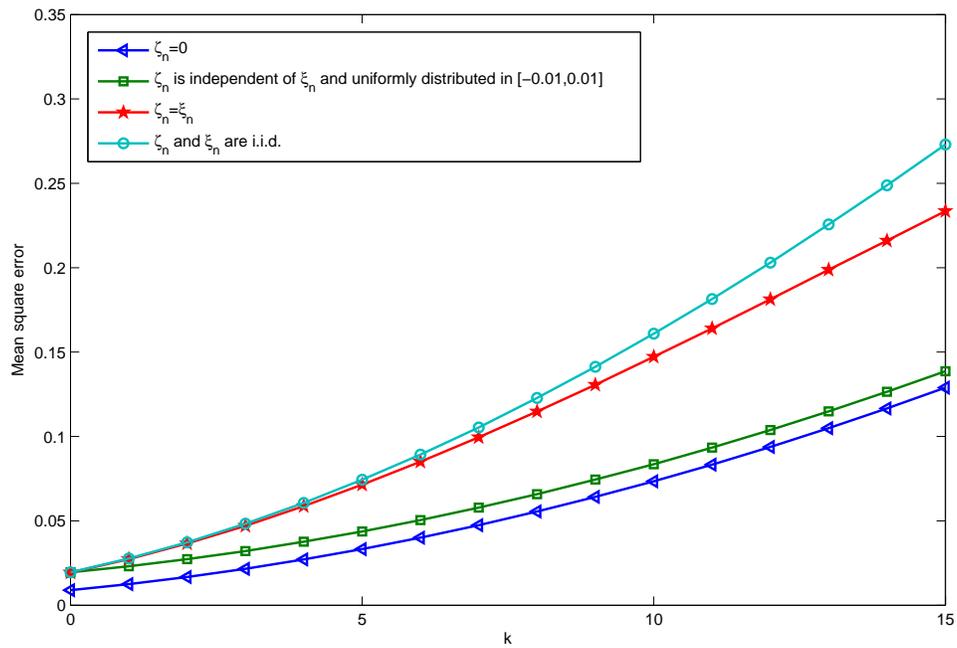}
\caption{Mean square error of the reconstruction when $\xi_n$ is uniformly distributed in the interval $[-0.01-0.002*k,0.01+0.002*k]$, where $k=0,1,\cdots,15$.}
 \label{fig13}
\end{figure}

\end{document}